\documentclass[11pt]{article}
\usepackage[margin=1.1in]{geometry}

\usepackage[utf8]{inputenc} 
\usepackage{booktabs}       
\usepackage{amsfonts}       
\usepackage{nicefrac}       
\usepackage{microtype}      
\usepackage{parskip}        
\usepackage[dvipsnames]{xcolor} 
\usepackage{amsmath, amsthm, mathtools, dsfont}
\usepackage{txfonts} 
\usepackage{color,graphicx}
\usepackage{url,hyperref}
\usepackage{enumerate}
\usepackage[ruled,linesnumbered, lined, noend]{algorithm2e}
\usepackage{csquotes}
\usepackage[T1]{fontenc} %
\usepackage{color-edits}
\addauthor[Kiarash]{kb}{cyan}
\addauthor[Todo]{td}{red}
\addauthor[suppress]{gh}{OliveGreen}
\usepackage[normalem]{ulem}
\usepackage{pifont}

\newcommand{\term}{\texttt}
\newcommand{\alg}{\term{alg}}

\newcommand{\EXP}{\mathbb{E}}
\newcommand{\PROB}{\textnormal{Pr}}

\renewcommand{\Pr}[1]{{\PROB \sbr{#1}}}

\newcommand{\Ex}[1]{{\EXP \sbr{#1}}}
\newcommand{\Exu}[2]{\ensuremath{\EXP_{#1}\left[#2\right]}}

\newcommand{\rbr}[1]{\left(#1\right)}
\newcommand{\abs}[1]{\left|#1\right|}
\newcommand{\sbr}[1]{\left[#1\right]}
\newcommand{\cbr}[1]{\left\{#1\right\}}
\newcommand{\EqComment}[1]{\text{\emph{(#1)}}}

\newtheorem{theorem}{Theorem}
\newtheorem{corollary}[theorem]{Corollary}
\newtheorem{lemma}[theorem]{Lemma}

\title{Pandora with Inaccurate Priors}

\author{
    Kiarash Banihashem\thanks{University of Maryland, \texttt{kiarash@umd.edu}} \and
    Xiang Chen\thanks{Adobe Research, \texttt{xiangche@adobe.com}} \and
    MohammadTaghi Hajiaghayi\thanks{University of Maryland, \texttt{hajiagha@cs.umd.edu}} \and
    Sungchul Kim\thanks{Adobe Research, \texttt{sukim@adobe.com}} \and
    Kanak Mahadik\thanks{Adobe Systems, \texttt{mahadik@adobe.com}} \and
    Ryan Rossi\thanks{Adobe Research, \texttt{ryrossi@adobe.com}} \and
    Tong Yu\thanks{Adobe Research, \texttt{tyu@adobe.com}}
}

\date{}

\begin{document}
\maketitle

\begin{abstract}
We investigate the role of inaccurate priors for the classical Pandora's box problem. In the classical Pandora's box problem we are given a set of boxes each with a known cost and an unknown value sampled from a known distribution.
    We investigate how inaccuracies in the  beliefs can affect existing algorithms.
    Specifically, we assume that the knowledge of the underlying distribution has a small error in the Kolmogorov distance, and study how this affects the utility obtained by the optimal algorithm.
\end{abstract}

\section{Introduction}

The Pandora's Box problem is a classic problem in decision theory. We are presented with $n$ boxes, where each box $i$ has an associated cost $c_i \in \mathbb{R}^{\geq 0}$ and a value $X_i$, a random variable drawn from the distribution $D_i$.
The distributions $D_1, \dots, D_n$ and the costs $c_1, \dots, c_n$ are known in advance, but the values $X_1, \dots, X_n$ remain unknown. Intuitively, the cost $c_i$ and the distribution $D_i$ can be thought of as being written on the "label" of the box, while the value $X_i$ is hidden "inside" the box.

The objective is to decide which boxes to open and then select a value from one of the opened boxes in a way that maximizes the net utility. We must pay the cost for every box we open but will only receive the value of the box we ultimately choose to keep.
Let $O$ denote the set of opened boxes and $a$ represent the box we decide to keep. The total utility is given by:
\begin{align*}
    X_a - \sum_{i \in O} c_i. 
\end{align*}
Note that we can only keep a box that we have opened, i.e., $a \in O$. Additionally, the selection of boxes can be done \emph{adaptively}; the decision to open a subsequent box may depend on the values revealed in previously opened boxes.

The problem has an elegant solution based on the concept of \emph{reserve price}. Let $\sigma_i(D_i)$ denote a threshold satisfying
\begin{align*} \Ex{\sbr{X_i - \sigma_i}_{+}} = c_i, \end{align*}
where $\sbr{\cdot}_{+}$ represents $\max\cbr{0, \cdot}$. The reserve price $\sigma_i$ is the threshold beyond which the expected gain from a box equals its cost of opening.
The optimal algorithm is to open boxes in descending order of their reserve prices, as formalized in the following lemma.

\begin{lemma}\label{lm:orig}
    Set $\kappa_i = \min \cbr{\sigma_i, X_i}$.
    The expected utility of any algorithm for the Pandora's box problem is $\Ex{\max_i \kappa_i}$.
    Furthermore, the following algorithm achieves this bound.
    Open boxes in descending order of $\sigma_i$ and stop when the next $\sigma_i$ is below the current largest value $X_i$ among the opened boxes. At that point, take the largest value among opened boxes.
\end{lemma}

The original formulation of the problem is attributed to Weitzman~\cite{weitzman1978optimal}.
More recently, Kleinberg, Waggoner, and
Weyl~\cite{kleinberg2016descending} study the application of thie problem in auctions and provide alternative proof of the optimal solution. Subsequent works have considered variations of the problem including with different settings including online arrival order~\cite{esfandiari2019online} correlated distributions~\cite{chawla2020pandora}. We refer to the survey by Beyhaghi and Cai~\cite{beyhaghi2024recent} for more details.

In this paper, we investigate how inaccurate knowledge of the distributions
can affect this problem. 
Specifically, we study the following problem. 
Assume that instead of the distribution $D_i$, we are given a distribution $D'_i$ which differs from $D_i$ by at most $\epsilon$ in Kolmogorov distance.
To what extent does the error in our knowledge of the distribution affect the utility of the algorithm?
This question has previously been studied for the prophet inequalities problem by D{\"u}tting and Kesselheim~\cite{dutting2019posted}, which serves as the main inspiration behind this works. To our knowledge we are the first to investigate inaccurate priors for the Pandora's box problem.

\section{Preliminaries and notation}
We use $X_i$ and $c_i$ to denote, respectively, the value and the cost of box $i$.
The value of each box is random and we assume it is sampled from some distribution, which we usually denote with $D_i$. 
We assume throughout that $D_i$ is supported over the set $[0, 1]$. 
Given a box with distribution $D$ and cost $c$, we refer to
$\sigma(D, c)$ as the reservation price for the box and define it as the value satisfying $\Exu{X \sim D}{\sbr{X - \sigma(D, c)}_+} = c$. We will drop the dependence on $D$ and $c$ when it is clear from context. 
We will assume that the distributions $D$ are continuous and as such, the threshold $\sigma$ exists. 
We use bold notation to denote vector version of the quantities above; e.g., $\mathbf{D} := \rbr{D_1, \dots, D_n}$.

Given a vector $\mathbf{\tau} = (\tau_1, \dots, \tau_n)$ for estimated reserved values for the boxes, we use
$\alg_{\tau}$
to denote
an algorithm which opens box in decreasing order of $\tau_i$ and stops when the maximum among already observed values exceeds the next value of $\tau_i$, at which point the algorithm accepts the highest value already seen (or $0$ if no box has been opened).
We use $W_{\tau}$ to denote the expected utility of this algorithm.
The expected utility implicity depends on the value distribution and the cost of each box, and we will often make this dependence explicit by writing either $W_{\mathbf{\tau}}(\mathbf{D})$ or
$W_{\mathbf{\tau}}(\mathbf{D}, \mathbf{c})$.
Note that, in general, we do not require $\mathbf{\tau} = \mathbf{\sigma}(\mathbf{D}, \mathbf{c})$. 
We will often abuse notation and use
$W_{\mathbf{\tau}}(x_1, \dots, x_n)$ for fixed values of $x_i$ to denote the utility when the box values are set to $x_i$.
We note that
$
\Ex{W_{\mathbf{\tau}}(\mathbf{D})}
= 
\Exu{X_1 \sim D_1, \dots, X_n \sim D_n} {W_{\mathbf{\tau}}(X_1, \dots, X_n)}
$.

Given two distributions $D, D'$, their Kolmogorov distance is defined as the supremum
of the difference between their
cumulative distribution functions;
letting $F_{D}$ denote the CDF of $D$, this can be formally written as
\begin{align*}
    d_{K}(D, D') =
    \sup_{z \in [0, 1]}\abs{F_{D}(z) - F_{D'}(z)}
    .
\end{align*}

\section{Main result}
Let $\mathbf{D}, \mathbf{D}'$ denote two different values for the input distributions, $\mathbf{c}$ denote the cost of opening the boxes and set
$\sigma = \sigma(\mathbf{D}, \mathbf{c})$ and
$\sigma' = \sigma(\mathbf{D}', \mathbf{c})$.
Our main result is the following theorem, which bounds how much the expected ulity of $\sigma$ is different under $\mathbf{D}$ and $\mathbf{D}'$.

\begin{theorem}\label{thm:kolomogrov_1}
    Assume that
    $d_{K}(D_i, D'_i) \le \epsilon$ for all $i$.
    Then
    \begin{align*}
        \abs{
        W_{\mathbf{\sigma}}(\mathbf{D}) 
        - 
        W_{\mathbf{\sigma}}(\mathbf{D}') 
        }
        \le O(n\epsilon).
    \end{align*}
\end{theorem}

While the above theorem bounds how much utility changes by changing the distribution, it does not say anything about how basing a strategy on the wrong distribution can affect its utility. The following corollary achieves exactly this, using the above theorem and the optimality of $\sigma$ and $\sigma'$ for their respective distributions.
\begin{corollary}
    If $d_{K}(D_i, D'_i) \le \epsilon$ for all $i$,
    \begin{align*}
         0 \le W_{\mathbf{\sigma}}(\mathbf{D})
         - W_{\mathbf{\sigma'}}(\mathbf{D})
         \le O(n\epsilon).
    \end{align*}
\end{corollary}
\begin{proof}
    The first inequality follows from the optimality of $\sigma$ for $\mathbf{D}$ (see Lemma~\ref{lm:orig}).
    We therefore focus on the second inequality. 
    We start by observing that, since $\sigma'_i$ denotes the reservations prices for the distribution $D'_i$,
    it leads to the optimal expected utility when the values in the boxes are sampled from $D'_i$. Formally, by Lemma~\ref{lm:orig},
    \begin{align}
        W_{\mathbf{\sigma'}}(\mathbf{D'})
        \ge 
        W_{\mathbf{\sigma}}(\mathbf{D'})
        \label{eq:unif_converge_1}
    \end{align}
    It follows that
    \begin{align*}
         W_{\mathbf{\sigma'}}(\mathbf{D})
         &\ge 
         W_{\mathbf{\sigma'}}(\mathbf{D'})
         - O(n\epsilon)
         &\EqComment{Theorem~\ref{thm:kolomogrov_1}}
         \\&\ge 
         W_{\mathbf{\sigma}}(\mathbf{D'})
         - O(n\epsilon)
         &\EqComment{Equation~\eqref{eq:unif_converge_1}}
         \\&\ge 
         W_{\mathbf{\sigma}}(\mathbf{D})
         - O(n\epsilon)
         &\EqComment{Theorem~\ref{thm:kolomogrov_1}}.
    \end{align*}
\end{proof}

We proceed to prove Theorem~\ref{thm:kolomogrov_1}.
Assume without loss of generality that the boxes are ordered in decreasing value of $\sigma_i$; i.e., 
$\sigma_1 \ge \sigma_2 \ge \dots, \ge \sigma_n$.

\begin{lemma}\label{lm:g_increas_contin}
    Fix $i$ and the values $x_1, \dots, x_{i-1}$.
    Define $g(x) := W_{\mathbf{\sigma}}(x_1, \dots, x_{i-1}, x, D_{i+1}, \dots, D_n)$.    
    The function $g(.)$ is increasing and $1$-lipschitz (and therefore continuous).
\end{lemma}
\begin{proof}
    The value of $x_{1}, \dots, x_{i-1}$ determine
    whether or not box $i$ is opened in the algorithm.
    If the box is never opened, then
    $g(x)$ is always a constant value and the claim is proved.
    We therefore assume that box $i$ is opened for the rest of the proof.
    
    Define $B(x) = \max \cbr{x_1, \dots, x_{i-1}, x}$. 
    Let $c_{\le i} := \sum_{j \le i} c_{j}$ denote the cost of opening boxes $1$ to $i$.
    We claim that
    \begin{align}
        W_{\mathbf{\sigma}}(x_1, \dots, x_{i-1}, x, D_{i+1}, \dots, D_n) = 
        \Exu{X_{i+1} \sim D_{i+1}, \dots, X_n \sim D_n}{\max \cbr{B(x), \max_{i+1 \le j\le n} \kappa_j}} - c_{\le i}.
        \label{eq:W_decompose_g}
    \end{align}
    Observe that this would finish the proof since
    the above function is clearly $1$-lipschitz and increasing in $x$.
    
    To prove this, consider the following variant of the problem.
    For all $j \ge i + 1$, instead of charging the algorithm $c_j$ before it opens the box,
    we charge it the excess value $[X_j - \sigma_j]_{+}$ after it opens the box.
    Note that the choice of which boxes to open does not change since we still use $\alg_{\sigma}$; as before we keep opening boxes in descending order of the reservation price and if the next box's reservation price exceeds
    the current maximum value,
    we accept the current maximum value.

    We claim that the
    expected utility of the algorithm
    would not change.
    This holds because, for each $j > i$,
    we have $\Ex{[X_j - \sigma_j]+} = c_j$ by definition of $\sigma$, which means that the expected cost of opening box $j$ does not change. 
    Since expectation is linear, the overall cost of opening boxes does not change in expectation. Since the value accepted by the algorithm and the boxes opened do no change, the overall expected utility stays the same. 
    
    We further alter the problem by assuming that the algorithm only pays the cost $[X_j - \sigma_j]_{+}$ for a box $j > i$ if it accepts box $j$. We claim that this also does not change the expected utility of the algorithm.
    Let $i'$ be the first box such that
    $X_{i'} > \sigma_{i'}$.
    Since $\sigma_j$ is decreasing in $j$, we have
    $X_{i'}> \sigma_{j}$ for all $j > i'$. Therefore, box $i$ is accepted, which means the algorithm is not charged $[X_j - \sigma_j]_{+}$ for $j > i$, finishing the proof.

    Let $\tilde{W}(x_1, x_{i-1}, x, D_{i+1}, \dots, D_n)$ denote
    the utility in the new model. By the above discussion, we have
    \begin{align*}
        \tilde{W}_{\sigma}(x_1, x_{i-1}, x, D_{i+1}, \dots, D_n)
        = 
        {W}_{\sigma}(x_1, x_{i-1}, x, D_{i+1}, \dots, D_n)
        .
    \end{align*}
    Next, we analyze $\tilde{W}(x_1, x_{i-1}, x, D_{i+1}, \dots, D_n)$.
    We will show that, for any value of $X_{i+1}, \dots, X_{n}$,
    the algorithm's utility is exactly
    \begin{align*}
        \max\cbr{B(x), \max_{i+1 \le j \le n} \kappa_j}- c_{\le i}
        .
    \end{align*}
    Taking expectations over $X_{i+1}, \dots, X_n$, this finishes the proof.
    To simplify notation, we define $X_j=x_j$ for $j < i$ and
    $X_i = x$.

    We first show that
    the algorithm's utility is at most the mentioned value.
    To prove this, observe that
    the algorithm already pays $c_{\le i}$ for opening boxes $1$ to $i$.
    If it accepts a box $j$ for $j \le i$, then it receives at most
    $B(x)$, and the claim is proved. Otherwise, it must accept a box $j \ge i + 1$, which means it receives
    $X_{j}$ but also pays $[X_j - \sigma_j]_{+}$. Therefore, it receives a net value of $\kappa_j \le \max_{\ell} \kappa_\ell$.

    We next show that
    the algorithm's utility is at least
    the mentioned value.
    We first observe that since box $i$ is opened, the algorithm always accepts some box $j \in [n]$. 
    The algorithm has already paid the cost $c_{\le i}$ so we will focus on analyzing the value of the box it accepts and the cost it may potentially pay for it.
    
    Let $r \ge i$ denote the last box that was opened by the algorithm.
    Since the algorithm chooses the opened box with the largest value,
    we have
    \begin{align}
        X_j \ge X_{\ell} \text{ for } \ell \in [r].
        \label{eq:jan7_1}
    \end{align}
    Additionally, since boxes $\ell > r$ were not opened (assuming such boxes exist), we have
    \begin{align}
        X_j \ge \sigma_\ell\text{ for } \ell > r
        .
        \label{eq:jan7_3}
    \end{align}

    If $j \le i$, then
    the algorithm receives
    $X_j$. This is at least $B(x)$ because of Equation~\eqref{eq:jan7_1} and the fact that $i \le r$. It is also at least $\kappa_\ell$ for any $\ell > i$;
    it is at least 
    $X_\ell$ if $\ell \le r$ because of Equation~\eqref{eq:jan7_1} and it is at least $\sigma_\ell$ if $\ell > r$ because of Equation~\eqref{eq:jan7_3}. Therefore, the claim is proved in this case.
    
    If $j \ge i + 1$, then the algorithm receives $X_j$ but also pays the extra cost 
    $[X_j - \sigma_j]_{+}$, which means it receives the net value
    $\kappa_j$.
    We therefore need to show
    $\kappa_j \ge X_\ell$
    for $\ell \le i$
    and $\kappa_j \ge \kappa_\ell$ for $\ell \ge i + 1$.
    The
    fact that box $j$ was opened means that
    \begin{align}
        \sigma_j \ge \max_{\ell < j} X_\ell
        .
        \label{eq:jan7_2}
    \end{align}
    If $\ell \le i$
    we have 
    $X_j \ge X_\ell$ because of Equation~\eqref{eq:jan7_1} and the fact that $\ell \le i \le r$.
    We also have $\sigma_j \ge X_\ell$ because of 
    Equation~\eqref{eq:jan7_2}
    and the assumption $j > i \ge \ell$.
    Therefore
    \begin{align*}
        \kappa_j
        = \min\{
        \sigma_j, X_j
        \}
        \ge X_\ell
        .
    \end{align*}
    
    If $i + 1 \le \ell \le j$, then 
    we have
    $X_j \ge X_\ell \ge \kappa_\ell$ by
    Equation~\eqref{eq:jan7_1}
    and 
    $\sigma_j \ge X_\ell \ge \kappa_\ell$ because 
    of Equation~\eqref{eq:jan7_2}. 
    If $j < \ell \le r$,
    we have
    $X_j \ge X_\ell \ge \kappa_\ell$ by
    Equation~\eqref{eq:jan7_1}
    and 
    $\sigma_j \ge \sigma_\ell \ge \kappa_\ell$
    because 
    the algorithm inspects in decreasing order of $\sigma$.
    Finally,
    if 
    $\ell > r$,
    then we have
    $X_j \ge \sigma_\ell \ge \kappa_\ell$ because of Equation~\eqref{eq:jan7_3}
    and $\sigma_j \ge \sigma_\ell \ge \kappa_\ell$ 
    because 
    the algorithm inspects in decreasing order of $\sigma$.

\end{proof}

\begin{lemma}\label{lm:jan14_2}
    Fix $i$ and $x_1, \dots, x_{i-1}$.
    \begin{align*}
        \abs{
        W_{\mathbf{\sigma}}(x_1, \dots, x_{i-1}, D'_i, D_{i+1}, \dots, D_n)
        - 
        W_{\mathbf{\sigma}}(x_1, \dots, x_{i-1}, D_i, D_{i+1}, \dots, D_n)
        }
        \le O(\epsilon)
    \end{align*}
\end{lemma}
\begin{proof}
    Sample the random variables $X_i$ and $X'_i$ from $D_i$ and $D'_i$ respectively.
    Define $g(x)$ as
    \begin{align*}
        W_{\mathbf{\sigma}}(x_1, \dots, X_{i-1}, x, D_{i+1}, \dots, D_n).
    \end{align*}
    We need to show that
    $\abs{\Ex{g(X) - g(X')}} \le O(\epsilon)$.

    Given Lemma~\ref{lm:g_increas_contin},
    for any value $t$, there exists a value $x(t) \in  [-\infty, +\infty]$ such that
    \begin{align*}
        \Pr{g(X) \ge t}
        = \Pr{ 
            X \ge x(t)
        }.
    \end{align*}
    Note we are including the values $+\infty$ and $-\infty$ to account for the possibility that the event always holds or it never holds.
    Specifically, according to Lemma~\ref{lm:g_increas_contin} the function
$W_{\mathbf{\sigma}}(x_1, \dots, X_{i-1}, x, D_{i+1}, \dots, D_n)$ increases in $x$, which means that the set of all $x$ for which it attains a value greater than $t$ has one of the following forms.
    \begin{itemize}
        \item It is the empty. In this case $x(t)=+\infty$.
        \item 
        It is the set $(-\infty, +\infty)$. In this case $x(t) = -\infty$.
        \item It is an interval of the form $[x(t), \infty)$. 
        \item It is an interval of the form $(x(t), \infty)$. This cannot happen however because the function is also continous, which means if it is at least $t$ for all values greater than $x(t)$, it is at least $t$ for $x(t)$ as well.
    \end{itemize}
    
    Similarly, 
    \begin{align*}
        \Pr{g(X') \ge t}
        = \Pr{ 
            X' \ge x(t)
        }.
    \end{align*}
    Therefore,
    $\Pr{g(X) \ge t}$.
    By assumption on $D_i, D'_i$, we have
    $\abs{\Pr{X \ge x(t)} - \Pr{X' \ge x(t)}} \le \epsilon$, which implies
    \begin{align}
        \abs{\Pr{g(X) \ge t} - \Pr{g(X') \ge t}} \le \epsilon
        \label{eq:jan14_1}
    \end{align}
    for all $t$.
    
    It follows that
    \begin{align*}
        \Ex{\abs{ g(X) - g(X') }} 
        &= 
        \int_{0}^{\infty} 
        \abs{\Pr{g(X) \ge t}
        - 
        \Pr{g(X') \ge t} }
        \; dt
    \end{align*}
    Note however that, since $D, D'$ are always in the set $[0, 1]$, we have
    $\Pr{g(X) > g(1)} = \Pr{g(X') > g(1)} = 0$ 
    and
    $\Pr{g(X) \ge g(0)} = \Pr{g(X') \ge g(0)} = 1$.
    Therefore, we can rewrite the above as
    \begin{align*}
        \int_{g(0)}^{g(1)} 
        \abs{\Pr{g(X) \ge t}
        - 
        \Pr{g(X') \ge t} }
        \; dt
        \le \epsilon (g(1) - g(0))
        \le \epsilon
    \end{align*}
    where the first inequality follows from Equation~\eqref{eq:jan14_1}
    and the second inequality follows from the fact that $g$ is $1$-lipschitz (see Lemma~\ref{lm:g_increas_contin}).
    
\end{proof}
Sampling $X'_1\sim D'_1, \dots, X'_{i} \sim D'_{i}$,
\begin{align*}
    &\abs{
        W_{\mathbf{\sigma}}(D'_1, \dots D'_{i-1}, D'_i, D_{i+1}, \dots, D_n)
        - 
        W_{\mathbf{\sigma}}(D'_1, \dots D'_{i-1}, D_i, D_{i+1}, \dots, D_n)
        }
    \\&=
    \abs{
        \Ex{
        W_{\mathbf{\sigma}}(X'_1, \dots X'_{i-1}, X'_i, D_{i+1}, \dots, D_n)
        - 
        W_{\mathbf{\sigma}}(X'_1, \dots X'_{i-1}, X_i, D_{i+1}, \dots, D_n)
        }
        }
    \\&\le 
    \Ex{
    \abs{
    W_{\mathbf{\sigma}}(X'_1, \dots X'_{i-1}, X'_i, D_{i+1}, \dots, D_n)
        - 
        W_{\mathbf{\sigma}}(X'_1, \dots X'_{i-1}, X_i, D_{i+1}, \dots, D_n)
        }
    }
    \\&\le 
    \Ex{O(\epsilon)},
\end{align*}
where the first inequality follows from Jensen and the second inequality follows from Lemma~\ref{lm:jan14_2}.

Summing over all $i$ we obtain Theorem~\ref{thm:kolomogrov_1}.

\bibliographystyle{abbrv}
\bibliography{ref}
\end{document}